\documentclass[12pt]{llncs}
\usepackage{fullpage,setspace}
\newcommand{\Oh}[1]
    {\ensuremath{\mathcal{O} \hspace{-0.5ex} \left( {#1} \right)}}
\newtheorem{open}{Open Problem}

\begin{document}

\title{On the Value of Multiple Read/Write Streams\\for Data Compression}
\author{Travis Gagie}
\institute{Department of Computer Science and Engineering\\
Aalto University, Finland\\
\email{travis.gagie@aalto.fi}}
\maketitle

\begin{abstract}
We study whether, when restricted to using polylogarithmic memory and polylogarithmic passes, we can achieve qualitatively better data compression with multiple read/write streams than we can with only one.  We first show how we can achieve universal compression using only one pass over one stream.  We then show that one stream is not sufficient for us to achieve good grammar-based compression.  Finally, we show that two streams are necessary and sufficient for us to achieve entropy-only bounds.
\end{abstract}

\section{Introduction} \label{sec:introduction}

Massive datasets seem to expand to fill the space available and, in situations where they no longer fit in memory and must be stored on disk, we may need new models and algorithms.  Grohe and Schweikardt~\cite{GS05} introduced read/write streams to model situations in which we want to process data using mainly sequential accesses to one or more disks.  As the name suggests, this model is like the streaming model (see, e.g.,~\cite{Mut05}) but, as is reasonable with datasets stored on disk, it allows us to make multiple passes over the data, change them and even use multiple streams (i.e., disks).  As Grohe and Schweikardt pointed out, sequential disk accesses are much faster than random accesses --- potentially bypassing the von Neumann bottleneck --- and using several disks in parallel can greatly reduce the amount of memory and the number of accesses needed.  For example, when sorting, we need the product of the memory and accesses to be at least linear when we use one disk~\cite{MP80,GKS07} but only polylogarithmic when we use two~\cite{CY91,GS05}.  Similar bounds have been proven for a number of other problems, such as checking set disjointness or equality; we refer readers to Schweikardt's survey~\cite{Sch07} of upper and lower bounds with one or more read/write streams, Heinrich and Schweikardt's paper~\cite{HS08} relating read/write streams to classic complexity theory, and Beame and Huynh's paper~\cite{BH08} on the value of multiple read/write streams for approximating frequency moments.

Since sorting is an important operation in some of the most powerful data compression algorithms, and compression is an important operation for reducing massive datasets to a more manageable size, we wondered whether extra streams could also help us achieve better compression.  In this paper we consider the problem of compressing a string $s$ of $n$ characters over an alphabet of size $\sigma$ when we are restricted to using \(\log^{\mathcal{O} (1)} n\) bits of memory and \(\log^{\mathcal{O} (1)} n\) passes over the data.  Throughout, we write $\log$ to mean $\log_2$ unless otherwise stated.  In Section~\ref{sec:universal}, we show how we can achieve universal compression using only one pass over one stream.  Our approach is to break the string into blocks and compress each block separately, similar to what is done in practice to compress large files.  Although this may not usually significantly worsen the compression itself, it may stop us from then building a fast compressed index (see~\cite{NM07} for a survey) unless we somehow combine the indexes for the blocks, or clustering by compression~\cite{CV05} (since concatenating files should not help us compress them better if we then break them into pieces again).  In Section~\ref{sec:grammar-based} we use a vaguely automata-theoretic argument to show one stream is not sufficient for us to achieve good grammar-based compression.  Of course, by `good' we mean here something stronger than universal compression: we want to build a context-free grammar that generates $s$ and only $s$ and whose size is nearly minimum.  In a paper with Gawrychowski~\cite{GG10} we showed that with constant memory and logarithmic passes over a constant number of streams, we can build a grammar whose size is at most quadratic in the minimum. Finally, in Section~\ref{sec:entropy-only} we show that two streams are necessary and sufficient for us to achieve entropy-only bounds.  Along the way, we show we need two streams to find strings' minimum periods or compute the Burrows-Wheeler Transform.  As far as we know, this is the first paper on compression with read/write streams, and among the first papers on compression in any streaming model; we hope the techniques we have used will prove to be of independent interest.

\section{Universal compression} \label{sec:universal}

An algorithm is called universal with respect to a class of sources if, when a string is drawn from any of those sources, the algorithm's redundancy per character approaches 0 with probability 1 as the length of the string grows.  The class most often considered, and which we consider in this section, is that of stationary, ergodic Markov sources (see, e.g.,~\cite{CT06}).  Since the $k$th-order empirical entropy \(H_k (s)\) of $s$ is the minimum self-information per character of $s$ with respect to a $k$th-order Markov source (see~\cite{Sav97}), an algorithm is universal if it stores any string $s$ in \(n H_k (s) + o (n)\) bits for any fixed $\sigma$ and $k$.  The $k$th-order empirical entropy of $s$ is also our expected uncertainty about a randomly chosen character of $s$ when given the $k$ preceding characters.  Specifically,
\[H_k (s) = \left\{ \begin{array}{ll}
    (1 / n) \sum_a \mathsf{occ} (a, s) \log \frac{n}{\mathsf{occ}(a, s)} &
    \mbox{if \(k = 0\),}\\[1ex]
    (1 / n) \sum_{|w| = k} |w_s| H_0 (w_s) &
    \mbox{otherwise,}
\end{array} \right.\]
where \(\mathsf{occ} (a, s)\) is the number of times character $a$ occurs in $s$, and $w_s$ is the concatenation of those characters immediately following occurrences of $k$-tuple $w$ in $s$.

In a previous paper~\cite{GM07b} we showed how to modify the well-known LZ77 compression algorithm~\cite{ZL77} to use sublinear memory while still storing $s$ in \(n H_k (s) + \Oh{n \log \log n / \log n}\) bits for any fixed $\sigma$ and $k$.  Our algorithm uses nearly linear memory and so does not fit into the model we consider in this paper, but we mention it here because it fits into some other streaming models (see, e.g.,~\cite{Mut05}) and, as far as we know, was the first compression algorithm to do so.  In the same paper we proved several lower bounds using ideas that eventually led to our lower bounds in Sections~\ref{sec:grammar-based} and~\ref{sec:entropy-only} of this paper.

\begin{theorem}[Gagie and Manzini, 2007] \label{thm:GM07b}
We can achieve universal compression using one pass over one stream and $\Oh{n / \log^2 n}$ bits of memory.
\end{theorem}

To achieve universal compression with only polylogarithmic memory, we use a algorithm due to Gupta, Grossi and Vitter~\cite{GGV08}.  Although they designed it for the RAM model, we can easily turn it into a streaming algorithm by processing $s$ in small blocks and compressing each block separately.

\begin{theorem}[Gupta, Grossi and Vitter, 2008] \label{thm:GGV08}
In the RAM model, we can store any string $s$ in \(n H_k (s) + \Oh{\sigma^k \log n}\) bits, for all $k$ simultaneously, using $\Oh{n}$ time.
\end{theorem}

\begin{corollary} \label{cor:universal}
We can achieve universal compression using one pass over one stream and $\Oh{\log^{1 + \epsilon} n}$ bits of memory.
\end{corollary}

\begin{proof}
We process $s$ in blocks of \(\log^\epsilon n\) characters, as follows: we read each block into memory, apply Theorem~\ref{thm:GGV08} to it, output the result, empty the memory, and move on to the next block.  (If $n$ is not given in advance, we increase the block size as we read more characters.)  Since Gupta, Grossi and Vitter's algorithm uses $\Oh{n}$ time in the RAM model, it uses $\Oh{n \log n}$ bits of memory and we use $\Oh{\log^{1 + \epsilon} n}$ bits of memory.  If the blocks are \(s_1, \ldots, s_b\), then we store all of them in a total of
\[\sum_{i = 1}^b \left( |s_i| H_k (s_i) + \Oh{\sigma^k \log \log n} \right)
\leq n H_k (s) + \Oh{\sigma^k n \log \log n / \log^\epsilon n}\]
bits for all $k$ simultaneously.  Therefore, for any fixed $\sigma$ and $k$, we store $s$ in \(n H_k (s) + o (n)\) bits. \qed
\end{proof}

A bound of \(n H_k (s) + \Oh{\sigma^k n \log \log n / \log^\epsilon n}\) bits is not very meaningful when $k$ is not fixed and grows as fast as \(\log \log n\), because the second term is \(\omega (n)\).  Notice, however, that Gupta et al.'s bound of \(n H_k (s) + \Oh{\sigma^k \log n}\) bits is also not very meaningful when \(k \geq \log n\), for the same reason.  As we will see in Section~\ref{sec:entropy-only}, it is possible for $s$ to be fairly incompressible but still to have \(H_k (s) = 0\) for \(k \geq \log n\).  It follows that, although we can prove bounds that hold for all $k$ simultaneously, those bounds cannot guarantee good compression in terms of \(H_k (s)\) when \(k \geq \log n\).

By using larger blocks --- and, thus, more memory --- we can reduce the $\Oh{\sigma^k n \log \log n / \log^\epsilon n}$ redundancy term in our analysis, allowing $k$ to grow faster than \(\log \log n\) while still having a meaningful bound.  Specifically, if we process $s$ in blocks of $c$ characters, then we use $\Oh{c \log n}$ bits of memory and achieve a redundancy term of $\Oh{\sigma^k n \log c\,/\,c}$, allowing $k$ to grow nearly as fast as \(\log_\sigma c\) while still having a meaningful bound.  We will show later, in Theorem~\ref{thm:redundancy lbound}, that this tradeoff is nearly optimal: if we use $m$ bits of memory and $p$ passes over one stream and our redundancy term is $\Oh{\sigma^k r}$, then \(m p r = \Omega (n / f (n))\) for any function $f$ that increases without bound.  It is not clear to us, however, whether we can modify Corollary~\ref{cor:universal} to take advantage of multiple passes.

\begin{open} \label{opn:tradeoff}
With multiple passes over one stream, can we achieve better bounds on the memory and redundancy than we can with one pass?
\end{open}

\section{Grammar-based compression} \label{sec:grammar-based}

Charikar et al.~\cite{CLL+05} and Rytter~\cite{Ryt03} independently showed how to build a nearly minimal context-free grammar {\sf APPROX} that generates $s$ and only $s$.  Specifically, their algorithms yield grammars that are an $\Oh{\log n}$ factor larger than the smallest such grammar {\sf OPT}, which has size \(\Omega (\log n)\) bits.

\begin{theorem}[Charikar et al., 2005; Rytter, 2003] \label{thm:grammar_RAM}
In the RAM model, we can approximate the smallest grammar with \(|\mathsf{APPROX}| = \Oh{|\mathsf{OPT}|^2}\) using $\Oh{n}$ time.
\end{theorem}

\noindent In this section we prove that, if we use only one stream, then in general our approximation must be superpolynomially larger than the smallest grammar.  Our idea is to show that periodic strings whose periods are asymptotically slightly larger than the product of the memory and passes, can be encoded as small grammars but, in general, cannot be compressed well by algorithms that use only one stream.  Our argument is based on the following two lemmas.

\begin{lemma} \label{lem:small grammar}
If $s$ has period $\ell$, then the size of the smallest grammar for that string is $\Oh{\ell \log \sigma + \log n \log \log n}$ bits.
\end{lemma}

\begin{proof}
Let $t$ be the repeated substring and $t'$ be the proper prefix of $t$ such that \(s = t^{\lfloor n / \ell \rfloor} t'\).  We can encode a unary string $X^{\lfloor n / \ell \rfloor}$ as a grammar $G_1$ with $\Oh{\log n}$ productions of total size $\Oh{\log n \log \log n}$ bits.  We can also encode $t$ and $t'$ as grammars $G_2$ and $G_3$ with $\Oh{\ell}$ productions of total size $\Oh{\ell \log \sigma}$ bits.  Suppose $S_1$, $S_2$ and $S_3$ are the start symbols of $G_1$, $G_2$ and $G_3$, respectively.  By combining those grammars and adding the productions \(S_0 \rightarrow S_1 S_3\) and \(X \rightarrow S_2\), we obtain a grammar with $\Oh{\ell + \log n}$ productions of total size $\Oh{\ell \log \sigma + \log n \log \log n}$ bits that maps $S_0$ to $s$. \qed
\end{proof}

\begin{lemma} \label{lem:substring}
Consider a lossless compression algorithm that uses only one stream, and a machine performing that algorithm.  We can compute any substring from
\begin{itemize}
\item its length;
\item for each pass, the machine's memory configurations when it reaches and leaves the part of the stream that initially holds that substring;
\item all the output the machine produces while over that part.
\end{itemize}
\end{lemma}

\begin{proof}
Let $t$ be the substring and assume, for the sake of a contradiction, that there exists another substring $t'$ with the same length that takes the machine between the same configurations while producing the same output.  Then we can substitute $t'$ for $t$ in $s$ without changing the machine's complete output, contrary to our specification that the compression be lossless. \qed
\end{proof}

Lemma~\ref{lem:substring} implies that, for any substring, the size of the output the machine produces while over the part of the stream that initially holds that substring, plus twice the product of the memory and passes (i.e., the number of bits needed to store the memory configurations), must be at least that substring's complexity.  Therefore, if a substring is not compressible by more than a constant factor (as is the case for most strings) and asymptotically larger than the product of the memory and passes, then the size of the output for that substring must be at least proportional to the substring's length.  In other words, the algorithm cannot take full advantage of similarities between substrings to achieve better compression.  In particular, if $s$ is periodic with a period that is asymptotically slightly larger than the product of the memory and passes, and $s$'s repeated substring is not compressible by more than a constant factor, then the algorithm's complete output must be \(\Omega (n)\) bits.  By Lemma~\ref{lem:small grammar}, however, the size of the smallest grammar that generates $s$ and only $s$ is bounded in terms of the period.

\begin{theorem} \label{thm:grammar-based}
With one stream, we cannot approximate the smallest grammar with \(|\mathsf{APPROX}| \leq |\mathsf{OPT}|^{\mathcal{O} (1)}\).
\end{theorem}

\begin{proof}
Suppose an algorithm uses only one stream, $m$ bits of memory and $p$ passes to compress $s$, with \(m p = \log^{\mathcal{O} (1)} n\), and consider a machine performing that algorithm.  Furthermore, suppose $s$ is binary and periodic with period \(m p \log n\) and its repeated substring $t$ is not compressible by more than a constant factor.  Lemma~\ref{lem:substring} implies that the machine's output while over a part of the stream that initially holds a copy of $t$, must be \(\Omega (m p \log n - m p) = \Omega (m p \log n)\).  Therefore, the machine's complete output must be \(\Omega (n)\) bits.  By Lemma~\ref{lem:small grammar}, however, the size of the smallest grammar that generates $s$ and only $s$ is \(\mathcal{O} (m p \log n + \log n \log \log n) \subset \log^{\mathcal{O} (1)} n\) bits.  Since \(n = \log^{\omega (1)} n\), the algorithm's complete output is superpolynomially larger than the smallest grammar. \qed
\end{proof}

As an aside, we note that a symmetric argument shows that, with only one stream, in general we cannot decode a string encoded as a small grammar.  To see why, instead of considering a part of the stream that initially holds a copy of the repeated substring $t$, consider a part that is initially blank and eventually holds a copy of $t$.  (Since $s$ is periodic and thus very compressible, its encoding takes up only a fraction of the space it eventually occupies when decompressed; without loss of generality, we can assume the rest is blank.)
An argument similar to the proof of Lemma~\ref{lem:substring} shows we can compute $t$ from the machine's memory configurations when it reaches and leaves that part, so the product of the memory and passes must again be greater than or equal to $t$'s complexity.

\begin{theorem} \label{thm:decoding}
With one stream, we cannot decompress strings encoded as small grammars.
\end{theorem}

Theorem~\ref{thm:grammar-based} also has the following corollary, which may be of independent interest.

\begin{corollary} \label{cor:find period-}
With one stream, we cannot find strings' minimum periods.
\end{corollary}

\begin{proof}
Consider the proof of Theorem~\ref{thm:grammar-based}.  Notice that, if we could find $s$'s minimum period, then we could store $s$ in \(\log^{\mathcal{O} (1)} n\) bits by writing $n$ and one copy of its repeated substring $t$.  It follows that we cannot find strings' minimum periods. \qed
\end{proof}

Corollary~\ref{cor:find period-} may at first seem to contradict work by Erg{\"u}n, Muthukrishnan and Sahinalp~\cite{EMS04}, who gave streaming algorithms for determining approximate periodicity.  Whereas we are concerned with strings which are truly periodic, however, they were concerned with strings in which the copies of the repeated substring can differ to some extent.  To see why this is an important difference, consider the simple case of checking whether $s$ has period \(n / 2\) (i.e., whether or not it is a square).  Suppose we know the two halves of $s$ are either identical or differ in exactly one position, and we want to determine whether $s$ truly has period \(n / 2\); then we must compare each corresponding pair of characters and, by a crossing-sequences argument (see, e.g.,~\cite{MP80} for details of a similar argument), this takes \(\Omega (n / m)\) passes.  Now suppose we care only whether the two halves of $s$ match only in nearly all positions; then we need compare only a few randomly chosen pairs to decide correctly with high probability.

\begin{theorem} \label{thm:check period-}
With one stream, we cannot even check strings' minimum periods.
\end{theorem}

In the conference version of this paper we left as an open problem proving whether or not multiple streams are useful for grammar-based compression.  As we noted in the introduction, in a subsequent paper with Gawrychowski~\cite{GG10} we showed that with constant memory and logarithmic passes over a constant number of streams, we can approximate the smallest grammar with \(|\mathsf{APPROX}| = \Oh{|\mathsf{OPT}|^2}\), answering our question affirmatively.

\section{Entropy-only bounds} \label{sec:entropy-only}

Kosaraju and Manzini~\cite{KM99} pointed out that proving an algorithm universal does not necessarily tell us much about how it behaves on low-entropy strings.  In other words, showing that an algorithm encodes $s$ in \(n H_k (s) + o (n)\) bits is not very informative when \(n H_k (s) = o (n)\).  For example, although the well-known LZ78 compression algorithm~\cite{ZL78} is universal, \(|\mathsf{LZ78} (1^n)| = \Omega (\sqrt{n})\) while \(n H_0 (1^n) = 0\).  To analyze how algorithms perform on low-entropy strings, we would like to get rid of the \(o (n)\) term and prove bounds that depend only on \(n H_k (s)\).  Unfortunately, this is impossible since, as the example above shows, even \(n H_0 (s)\) can be 0 for arbitrarily long strings.

It is not hard to show that only unary strings have \(H_0 (s) = 0\).  For \(k \geq 1\), recall that \(H_k (s) = (1 / n) \sum_{|w| = k} |w_s| H_0 (w_s)\).  Therefore, \(H_k (s) = 0\) if and only if each distinct $k$-tuple $w$ in $s$ is always followed by the same distinct character.  This is because, if a $w$ is always followed by the same distinct character, then $w_s$ is unary, \(H_0 (w_s) = 0\) and $w$ contributes nothing to the sum in the formula.  Manzini~\cite{Man01} defined the $k$th-order modified empirical entropy \(H_k^* (s)\) such that each context $w$ contributes at least \(\lfloor \log |w_s| \rfloor + 1\) to the sum.  Because modified empirical entropy is more complicated than empirical entropy --- e.g., it allows for variable-length contexts --- we refer readers to Manzini's paper for the full definition.  In our proofs in this paper, we use only the fact that
\[n H_k (s) \leq n H_k^* (s) \leq n H_k (s) + \Oh{\sigma^k \log n}\,.\]

Manzini showed that, for some algorithms and all $k$ simultaneously, it is possible to bound the encoding's length in terms of only \(n H_k^* (s)\) and a constant $g_k$ that depends only on $\sigma$ and $k$; he called such bounds `entropy-only'.  In particular, he showed that an algorithm based on the Burrows-Wheeler Transform (BWT)~\cite{BW94} stores any string $s$ in at most \((5 + \epsilon) n H_k^* (s) + \log n + g_k\) bits for all $k$ simultaneously (since \(n H_k^* (s) \geq \log (n - k)\), we could remove the \(\log n\) term by adding 1 to the coefficient \(5 + \epsilon\)).

\begin{theorem}[Manzini, 2001] \label{thm:Man01}
Using the BWT, move-to-front coding, run-length coding and arithmetic coding, we can achieve an entropy-only bound.
\end{theorem}

\noindent The BWT sorts the characters in a string into the lexicographical order of the suffixes that immediately follow them.  When using the BWT for compression, it is customary to append a special character \$ that is lexicographically less than any in the alphabet.  For a more thorough description of the BWT, we again refer readers to Manzini's paper.  In this section we first show how we can compute and invert the BWT with two streams and, thus, achieve entropy-only bounds.  We then show that we cannot achieve entropy-only bounds with only one stream.  In other words, two streams are necessary and sufficient for us to achieve entropy-only bounds.

One of the most common ways to compute the BWT is by building a suffix array.  In his PhD thesis, Ruhl introduced the StreamSort model~\cite{Ruh03,ADRR04}, which is similar to the read/write streams model with one stream, except that it has an extra primitive that sorts the stream in one pass.  Among other things, he showed how to build a suffix array efficiently in this model.

\begin{theorem}[Ruhl, 2003] \label{thm:Ruh03}
In the StreamSort model, we can build a suffix array using $\Oh{\log n}$ bits of memory and $\Oh{\log n}$ passes.
\end{theorem}

\begin{corollary} \label{cor:BWT+}
With two streams, we can compute the BWT using $\Oh{\log n}$ bits of memory and $\Oh{\log^2 n}$ passes.
\end{corollary}

\begin{proof}
We can compute the BWT in the StreamSort model by appending \$ to $s$, building a suffix array, and replacing each value $i$ in the array by the \((i - 1)\)st character in $s$ (replacing either 0 or 1 by \$, depending on where we start counting).  This takes $\Oh{\log n}$ bits of memory and $\Oh{\log n}$ passes.  Since we can sort with two streams using $\Oh{\log n}$ bits memory and $\Oh{\log n}$ passes (see, e.g.,~\cite{Sch07}), it follows that we can compute the BWT using $\Oh{\log n}$ bits of memory and $\Oh{\log^2 n}$ passes. \qed
\end{proof}

We note as an aside that, once we have the suffix array for a periodic string, we can easily find its minimum period.  To see why, suppose $s$ has minimum period $\ell$, and consider the suffix $u$ of $s$ that starts in position \(\ell + 1\).  The longest common prefix of $s$ and $u$ has length \(n - \ell\), which is maximum; if another suffix $v$ shared a longer common prefix with $s$, then $s$ would have period \(n - |v| < \ell\).  It follows that, if the first position in the suffix array contains $i$, then the \((\ell + 1)\)st position contains \(i - 1\) (assuming $s$ terminates with \$, so $u$ is lexicographically less than $s$).  With two streams we can easily find the position \(\ell + 1\) that contains \(i - 1\) and then check that $s$ is indeed periodic with period $\ell$.

\begin{corollary} \label{cor:period+}
With two streams, we can compute a string's minimum period using $\Oh{\log n}$ bits and $\Oh{\log^2 n}$ passes.
\end{corollary}

Now suppose we are given a permutation $\pi$ on \(n + 1\) elements as a list \(\pi (1), \ldots, \pi (n + 1)\), and asked to rank it, i.e., to compute the list \(\pi^0 (1), \ldots, \pi^n (1)\).  This problem is a special case of list ranking (see, e.g.,~\cite{ABD+07}) and has a surprisingly long history.  For example, Knuth~\cite[Solution 24]{Knu98} described an algorithm, which he attributed to Hardy, for ranking a permutation with two tapes.  More recently, Bird and Mu~\cite{BM04} showed how to invert the BWT by ranking a permutation.  Therefore, reinterpreting Hardy's result in terms of the read/write streams model gives us the following bounds.

\begin{theorem}[Hardy, c.~1967] \label{thm:Har67}
With two streams, we can rank a permutation using $\Oh{\log n}$ bits of memory and $\Oh{\log^2 n}$ passes.
\end{theorem}

\begin{corollary} \label{cor:BWT-}
With two streams, we can invert the BWT using $\Oh{\log n}$ bits of memory and $\Oh{\log^2 n}$ passes.
\end{corollary}

\begin{proof}
The BWT has the property that, if a character is the $i$th in \(\mathsf{BWT} (s)\), then its successor in $s$ is the lexicographically $i$th in \(\mathsf{BWT} (s)\) (breaking ties by order of appearance).  Therefore, we can invert the BWT by replacing each character by its lexicographic rank, ranking the resulting permutation, replacing each value $i$ by the $i$th character of \(\mathsf{BWT} (s)\), and rotating the string until \$ is at the end.  This takes $\Oh{\log n}$ memory and $\Oh{\log^2 n}$ passes. \qed
\end{proof}

Since we can compute and invert move-to-front, run-length and arithmetic coding using $\Oh{\log n}$ bits of memory and $\Oh{1}$ passes over one stream, by combining Theorem~\ref{thm:Man01} and Corollaries~\ref{cor:BWT+} and~\ref{cor:BWT-} we obtain the following theorem.

\begin{theorem} \label{thm:entropy-only ubound}
With two streams, we can achieve an entropy-only bound using $\Oh{\log n}$ bits of memory and $\Oh{\log^2 n}$ passes.
\end{theorem}

It follows from Theorem~\ref{thm:entropy-only ubound} and a result by Hernich and Schweikardt~\cite{HS08} that we can achieve an entropy-only bound using $\Oh{1}$ bits of memory, $\Oh{\log^3 n}$ passes and four streams.  It follows from their theorem below that, with more streams, we can even reduce the number of passes to $\Oh{\log n}$.

\begin{theorem}[Hernich and Schweikardt, 2008] \label{thm:HS08}
If we can solve a problem with logarithmic work space, then we can solve it using $\Oh{1}$ bits of memory and $\Oh{\log n}$ passes over $\Oh{1}$ streams.
\end{theorem}

\begin{corollary} \label{cor:HS08}
With $\Oh{1}$ streams, we can achieve an entropy-only bound using $\Oh{1}$ bits of memory and $\Oh{\log n}$ passes.
\end{corollary}

\begin{proof}
To compute the $i$th character of \(\mathsf{BWT} (s)\), we find the $i$th lexicographically largest suffix.  To find this suffix, we loop though all the suffixes and, for each, count how many other suffixes are lexicographically less.  Comparing two suffixes character by character takes $\Oh{n^2}$ time, so we use a total of $\Oh{n^4}$ time; it does not matter now how much time we use, however, just that we need only a constant number of $\Oh{\log n}$-bit counters.  Since we can compute the BWT with logarithmic work space, it follows from Theorem~\ref{thm:HS08} that we can compute it --- and thereby achieve an entropy-only bound --- with $\Oh{1}$ bits of memory and $\Oh{\log n}$ passes over $\Oh{1}$ streams. \qed
\end{proof}

Although we have not been able to prove an \(\Omega (\log n)\) lower bound on the number of passes needed to achieve an entropy-only bound with $\Oh{1}$ streams, we have been able to prove such a bound for computing the BWT.  Our idea is to reduce sorting to the BWT, since Grohe and Schweikardt~\cite{GS05} showed we cannot sort $n$ numbers with \(o (\log n)\) passes over $\Oh{1}$ streams.  It is trivial, of course, to reduce sorting to the BWT if the alphabet is large enough --- e.g., linear in $n$ --- but our reduction is to the more reasonable problem of computing the BWT of a ternary string.

\begin{theorem} \label{thm:reduction}
With $\Oh{1}$ streams, we cannot compute the BWT using \(o (\log n)\) passes.
\end{theorem}

\begin{proof}
Suppose we are given a sequence of $n$ numbers \(x_1, \ldots, x_n\), each of \(2 \log n\) bits.  Grohe and Schweikardt showed we cannot generally sort such a sequence using \(o (\log n)\) passes over $\Oh{1}$ tapes.  We now use \(o (\log n)\) passes to turn \(x_1, \ldots, x_n\) into a ternary string $s$ such that, by calculating \(\mathsf{BWT} (s)\), we sort \(x_1, \ldots, x_n\).  It follows from this reduction that we cannot compute the BWT using \(o (\log n)\) passes, either.

With one pass, \(O (\log n)\) bits of memory and two tapes, for \(1 \leq i \leq n\) and \(1 \leq j \leq 2 \log n\), we replace the $j$th bit \(x_i [j]\) of $x_i$ by \(x_i [j]\ 2\ x_i\ i\ j\), writing 2 as a single character, $x_i$ in \(2 \log n\) bits, $i$ in \(\log n\) bits and $j$ in \(\log \log n + 1\) bits; the resulting string $s$ is of length \(2 n \log n (3 \log n + \log \log n + 2)\).  The only characters followed by 2s in $s$ are the bits at the beginning of replacement phrases, so the last \(2 n \log n\) characters of \(\mathsf{BWT} (s)\) are the bits of \(x_1, \ldots, x_n\); moreover, since the lexicographic order of equal-length binary strings is the same as their numeric order, the \(x_i [j]\) bits will be arranged by the $x_i$ values, with ties broken by the $i$ values (so if \(x_i = x_{i'}\) with \(i < i'\), then every \(x_i [j]\) comes before every \(x_{i'} [j']\)) and further ties broken by the $j$ values; therefore, the last \(2 n \log n\) bits of the transformed string are \(x_1, \ldots, x_n\) in sorted order.
\qed
\end{proof}

To show we need at least two streams to achieve entropy-only bounds, we use De Bruijn cycles in a proof similar to the one for Theorem~\ref{thm:grammar-based}.  A $\sigma$-ary De Bruijn cycle of order $k$ is a cyclic sequence in which every possible $k$-tuple appears exactly once.  For example, Figure~\ref{fig:cycles} shows binary De Bruijn cycles of orders 3 and 4.  Our argument this time is based on Lemma~\ref{lem:substring} and the results below about De Bruijn cycles.  We note as a historical aside that Theorem~\ref{thm:AB51} was first proven for the binary case in 1894 by Flye Sainte-Marie~\cite{Fly94}, but his result was later forgotten; De Bruijn~\cite{Bru46} gave a similar proof for that case in 1946, then in 1951 he and Van Aardenne-Ehrenfest~\cite{AB51} proved the general version we state here.

\begin{figure}[t]
\begin{center}
\begin{tabular}{c@{\hspace{20ex}}c}
\begin{tabular}{cccc}
& 0 & 0 &\\
1 &&& 0\\
0 &&& 1\\
& 1 & 1 &
\end{tabular}
& \begin{tabular}{cccccc}
1 & 0 & 0 & 0 & 0 & 1\\
1 &&&&& 0\\
1 &&&&& 0\\
1 & 0 & 1 & 0 & 1 & 1
\end{tabular}
\end{tabular}
\caption{Examples of binary De Bruijn cycles of orders 3 and 4.}
\label{fig:cycles}
\end{center}
\end{figure}

\begin{lemma} \label{lem:entropy}
If \(s \in d^*\) for some binary $\sigma$-ary De Bruijn cycle $d$ of order $k$, then \(n H_k^* (s) = \Oh{\sigma^k \log n}\).
\end{lemma}

\begin{proof}
By definition, each distinct $k$-tuple is always followed by the same distinct character; therefore, \(n H_k (s) = 0\) and \(n H_k^* (s) = \Oh{\sigma^k \log n}\). \qed
\end{proof}

\begin{theorem}[Van Aardenne-Ehrenfest and De Bruijn, 1951] \label{thm:AB51}
There are \(\left( \sigma!^{\sigma^{k - 1}} / \sigma^k \right)\) $\sigma$-ary De Bruijn cycles of order $k$.
\end{theorem}

\begin{corollary} \label{cor:cycle lbound}
We cannot store most $k$th-order De Bruijn cycles in \(o (\sigma^k \log \sigma)\) bits.
\end{corollary}

\begin{proof}
By Stirling's Formula, \(\log \left( \sigma!^{\sigma^{k - 1}} / \sigma^k \right) = \Theta (\sigma^k \log \sigma)\). \qed
\end{proof}

Since there are $\sigma^k$ possible $k$-tuples, $k$th-order De Bruijn cycles have length $\sigma^k$, so Corollary~\ref{cor:cycle lbound} means that we cannot compress most De Bruijn cycles by more than a constant factor.  Therefore, we can prove a lower bound similar to Theorem~\ref{thm:grammar-based} by supposing that $s$'s repeated substring is a De Bruijn cycle, then using Lemma~\ref{lem:entropy} instead of Lemma~\ref{lem:small grammar}.

\begin{theorem} \label{thm:entropy-only lbound}
With one stream, we cannot achieve an entropy-only bound.
\end{theorem}

\begin{proof}
As in the proof of Theorem~\ref{thm:grammar-based}, suppose an algorithm uses only one stream, $m$ bits of memory and $p$ passes to compress $s$, with \(m p = \log^{\mathcal{O} (1)} n\), and consider a machine performing that algorithm.  This time, however, suppose $s$ is binary and periodic with period $m p\,f (n)$, where \(f (n) = \Oh{\log n}\) is a function that increases without bound; furthermore, suppose $s$'s repeated substring $t$ is a $k$th-order De Bruijn cycle, \(k = \log (m p\,f (n))\), that is not compressible by more than a constant factor. Lemma~\ref{lem:substring} implies that the machine's output while over a part of the stream that initially holds a copy of $t$, must be \(\Omega (m p\,f (n) - m p) = \Omega (m p\,f (n))\).  Therefore, the machine's complete output must be \(\Omega (n)\) bits.  By Lemma~\ref{lem:entropy}, however, \(n H_k^* (s) = \Oh{2^k \log n} = \Oh{m p\,f (n) \log n} \subset \log^{\mathcal{O} (1)} n\). \qed
\end{proof}

Recall that in Section~\ref{sec:universal} we asserted the following claim, which we are now ready to prove.

\begin{theorem} \label{thm:redundancy lbound}
If we use $m$ bits of memory and $p$ passes over one stream and achieve universal compression with an \(\Oh{\sigma^k r}\) redundancy term, for all $k$ simultaneously, then \(m p r = \Omega (n / f (n))\) for any function $f$ that increases without bound.
\end{theorem}

\begin{proof}
Consider the proof of Theorem~\ref{thm:entropy-only lbound}: \(n H_k (s) = 0\) but we must output \(\Omega (n)\) bits, so \(r = \Omega (n / \sigma^k) = \Omega (n / (m p\,f (n)))\). \qed
\end{proof}

Notice Theorem~\ref{thm:entropy-only lbound} also implies a lower bound for computing the BWT: if we could compute the BWT with one stream then, since we can compute move-to-front, run-length and arithmetic coding using $\Oh{\log n}$ bits of memory and $\Oh{1}$ passes over one stream, we could thus achieve an entropy-only bound with one stream, contradicting Theorem~\ref{thm:entropy-only lbound}.

\begin{corollary} \label{cor:BWT_lbound}
With one stream, we cannot compute the BWT.
\end{corollary}

In the conference version of this paper~\cite{Gag09} we closed with a brief discussion of three entropy-only bounds that we proved with Manzini~\cite{GM07a}.  Our first bound was an improved analysis of the BWT followed by move-to-front, run-length and arithmetic coding (which lowered the coefficient from \(5 + \epsilon\) to \(4.4 + \epsilon\)), but our other bounds (one of which had a coefficient of \(2.69 + \epsilon\)) were analyses of the BWT followed by algorithms which we were not sure could be implemented with $\Oh{1}$ streams.  We now realize that, since both of these other algorithms can be computed with logarithmic work space, it follows from Theorem~\ref{thm:HS08} that they can indeed be computed with $\Oh{1}$ streams.

After having proven that we cannot compute the BWT with one stream, we promptly start working with Ferragina and Manzini on a practical algorithm~\cite{FGM12} that does exactly that.  However, that algorithm does not fit into the streaming models we have considered in this paper; in particular, the product of the internal memory and passes there is $\Oh{n \log n}$ bits, but we use only $n$ bits of workspace on the disk.  The existence of a practical algorithm for computing the BWT in external memory raises the question of whether we can query BWT-based compressed indexes quickly in external memory.  Chien et al.~\cite{CHSV08} proved lower bounds for indexed pattern matching in the external-memory model, but that model allows does not distinguish between sequential and random access to blocks.  The read/write-streams model is also inappropriate for analyzing the complexity of this task, since we can trivially use only one pass over one stream if we leave the text uncompressed and scan it all with a classic sequential pattern-matching algorithm.  Orlandi and Venturini~\cite{OV11} recently showed how we can store a sample of the BWT that lets us estimate what parts of the full BWT we need to read in order to answer a query.  If we modify their data structure slightly, we can make it recursive; i.e., with a smaller sample we can estimate what parts of the sample we need to read in order to estimate what parts of the full BWT we need to read.  Suppose we store on disk a set of samples whose sizes increase exponentially, finishing with the BWT itself.  We use each sample in turn to estimate what parts of the next sample we need to read, then read them into internal memory using only one pass over the next sample.  This increases the size of the whole index only slightly and lets us answer queries by reading few blocks and in the order they appear on disk.  We are currently working to optimize and implement this idea.

\section*{Acknowledgments}

Many thanks to Ferdinando Cicalese, Paolo Ferragina, Pawe{\l} Gawrychowski, Roberto Grossi, Ankur Gupta, Andre Hernich, Giovanni Manzini, Jens Stoye and Rossano Venturini, for helpful discussions.  This research was done while the author was at the University of Eastern Piedmont in Alessandria, Italy, supported by the Italy-Israel FIRB Project ``Pattern Discovery Algorithms in Discrete Structures, with Applications to Bioinformatics'', and at Bielefeld University, Germany, supported by the Sofja Kovalevskaja Award from the Alexander von Humboldt Foundation and the German Federal Ministry of Education and Research.

\bibliographystyle{plain}
\bibliography{ahlswede}

\end{document}